\documentclass[a4paper]{article}

\usepackage[british]{babel}
\usepackage[T1]{fontenc}
\usepackage[ansinew]{inputenc}
\usepackage{array}
\usepackage{a4}
\usepackage{a4wide}
\usepackage{Common/theorems}
\usepackage{Common/prooftree}

\usepackage{amsmath,amssymb,stmaryrd}

\long\def\ignore#1{\relax}

\newcommand\struto[1][15pt]{{\raise #1 \hbox{\strut}}}%
\newcommand\strutb[1][15pt]{{\raise-#1 \hbox{\strut}}}%

\newcommand\upline{\hline\struto[12pt]}
\newcommand\midline[1][4]{\\[+ #1pt]\hline\struto[#1pt]}

\newcommand\downline[1][12]{\\[+ #1pt]\hline}

\newcommand\pushright[1]{{              
    \parfillskip=0pt            
    \widowpenalty=1000000000         
    \displaywidowpenalty=10000  
    \finalhyphendemerits=0      
    \leavevmode                 
    \unskip                     
    \nobreak                    
    \hfil                       
    \penalty50                  
    \hskip.2em                  
    \null                       
    \hfill                      
    {#1}                        
    \par}}                      

\makeatletter

\renewcommand\[[1][1]{\par\removelastskip\vskip#1pt\vbox\bgroup\null\hfill$}
\renewcommand\][1][3]{$\hfill\null\egroup\par\removelastskip\vskip#1pt\noindent}

\newenvironment{centre}{\par\vbox\bgroup\null\hfill}{\hfill\null\egroup\par\vskip2pt}

\newcommand{\eqdef}{:=\ }
\newcommand{\recdef}{::=\ }

\newcommand{\Gam}{\Gamma}

\newcommand{\Del}{\Delta}

\newcommand\mathFomega{F_\omega}
\newcommand\Fomega{\ifmmode\mathFomega\else$\mathFomega$\fi}
\newcommand\mathFomegaC{F_\omega^{\mathcal C}}
\newcommand\FomegaC{\ifmmode\mathFomegaC\else$\mathFomegaC$\fi}
\newcommand\mathDNE{\mathrm{DNE}}
\newcommand\DNE{\ifmmode\mathDNE\else$\mathDNE$\fi}

\newcommand{\ie}{i.e.~}
\newcommand{\eg}{e.g.~}

\def\url#1#2{\texttt{#1}}

\newcommand\monthdisplay[1]{}

\newcommand{\multiset}[1]{\{ \!\! \{ #1\} \!\! \} }
\newcommand{\set}[1]{\{ #1\} }

\ignore{

}

\newcommand{\sep}{\mbox{$\;|\;$}}    

\newcommand{\seqg}[3]{\mbox{$\ {#1}_{#2}^{#3}\ $}}

\newcommand{\seqf}[2][]{\seqg{\vdash}{#1}{#2}}

\newcommand{\seqTh}[2][]{\seqg{\vdash_{\mathcal{T}}}{#1}{#2}}
\newcommand{\seqThD}[2][]{\seqg{\vdash_{\mathcal{T}}^{\mathcal P}}{#1}{#2}} 

\newcommand{\seq}[1][]{\seqf[#1]{}}

\newcommand{\seqt}[1][]{\seqTh[#1]{}}
\newcommand{\seqtD}[1][]{\seqThD[#1]{}}

\newcommand\DerNeg[4]{{#1}   \seqf{#4}   {#2}}

\newcommand\DerPosTh[5][{\mathcal T}]{{#2}   \seqf[#1]{#5}   {[#3]}}
\newcommand\DerNegTh[5][{\mathcal T}]{{#2}   \seqf[#1]{#5}   {#3}}

\newcommand\DerDPLL[2] {{#1}\mbox{;}{#2}\seq}

\newcommand\DerDPLLTh[2] {{#1}\mbox{;}{#2}\seqt}
\newcommand\DerDPLLThD[1] {{#1}\seqtD}  

\newcommand\daggerL{\raise3pt\hbox{\rotatebox{-40}{$\dagger$}}}
\newcommand\daggerR{\raise0pt\hbox{\rotatebox{40}{$\dagger$}}}

\newcommand\andP{{\wedge^+}}
\newcommand\andN{{\wedge^-}}
\newcommand\orP{{\vee^+}}
\newcommand\orN{{\vee^-}}

\newcommand\FA[2]{\forall #1 #2}

\newcommand{\non}[1]{{#1}^{\perp}}

\newcommand\Theory[2][\mathcal T]{{#1}(#2)}

\newcommand\LKDPLLc{\textsf{LK}${}^c_{\textsf{DPLL}}$}

\newcommand\LKDPLLTh[1][\mathcal T]{\textsf{LK}${}_{\textsf{DPLL}}$($#1$)}
\newcommand\LKDPLLThc[1][\mathcal T]{\textsf{LK}${}^c_{\textsf{DPLL}}$($#1$)}
\newcommand\LKDPLLThp[1][\mathcal T]{\textsf{LK}${}_{\textsf{DPLL}^{+}}$($#1$)}

\newcommand\LKp{\textsf{LK}$^p$}

\newcommand\LKThp[1][\mathcal T]{\textsf{LK}$^p$($#1$)}
\newcommand\size[1]{|#1|}

\newcommand\emptyenv{()}
\newcommand\backpoints[1]{[#1]}
\newcommand\backstrict[1]{\llbracket #1\rrbracket}
\newcommand\forget[1]{|#1|}

\newcommand\unsat{\textsf{UNSAT}}

\newcommand\DPLL{\textsf{DPLL}}

\newcommand\DPLLTh{\textsf{DPLL}($\mathcal T$)}
\newcommand\DPLLbjTh{\textsf{DPLL}$_{bj}$($\mathcal T$)}

\newcommand\mSat[1]{\textsf{Sat} (#1)}

\newcommand\atm[1]{\textsf{lit} (#1)}

\newcommand\nSat[2]{\textsf{Sat}_{#1}{(#2)}}

\begin{document}
\title{Two simulations about \DPLLTh}

\author{Mahfuza Farooque${}^1$, Stéphane Lengrand${}^{1,2}$ and Assia Mahboubi${}^3$\\[15pt]
  ${}^1$ CNRS\\
  ${}^2$ Ecole Polytechnique\\
  ${}^3$ Microsoft Research - INRIA Joint Centre\\[15pt]
  Project PSI: ``Proof Search control in Interaction with domain-specific methods''\\
  ANR-09-JCJC-0006
}

\maketitle
\abstract{
  In this paper we relate different formulations of the \DPLLTh\ procedure.

  The first formulation is that of~\cite{Nieuwenhuis06} based on a system of rewrite rules, which we denote \DPLLTh.

  The second formulation is an inference system of~\cite{Tin-JELIA-02}, which we denote \LKDPLLTh.

  The third formulation is the application of a standard proof-search mechanism in a sequent calculus \LKThp\ introduced here.

  We formalise an encoding from \DPLLTh\ to \LKDPLLTh\ that was, to our knowledge, never explicitly given and, in the case where \DPLLTh\ is extended with backjumping and Lemma learning, never even implicitly given. 

  We also formalise an encoding from \LKDPLLTh\ to \LKThp, building on Ivan Gazeau's previous work: we extend his work in that we handle the ``-modulo-Theory'' aspect of SAT-modulo-theory, by extending the sequent calculus to allow calls to a theory solver (seen as a blackbox). We also extend his work in that we handle advanced features of DPLL such as backjumping and Lemma learning, etc.

Finally, we refine the approach by starting to formalise quantitative aspects of the simulations: the complexity is preserved (nunber of steps to build complete proofs). Other aspects remain to be formalised (non-determinism of the search / width of search space).
}

\tableofcontents

\newpage
\section {Encoding DPLL($\mathcal T$) in \LKDPLLTh}
In this section we encode \DPLLTh\ in \LKDPLLTh.

Note that there exist different variants of \DPLLTh. We first consider the basic version which is equipped with backtracking.
This formalises ideas presented in~\cite{Tin-JELIA-02}.

Then we enhance the encoding to the enhanced version of \DPLLTh\ with backjumping, a generalised version of backtracking.

The main gap between \DPLLTh\ and an inference system such as \LKDPLLTh\ is the fact that a (successful) \DPLLTh\ run is a rewrite sequence finishing with the state \unsat, while a (successful) proof-search run is (/ produces) a proof tree. Roughly speaking, the \DPLLTh\ procedure implements the depth-first search of the corresponding tree.

\subsection{Preliminaries: \LKDPLLTh\ and its properties}

%
%

\begin{definition}[The system \LKDPLLTh]
  \emph{Clauses} are finite disjunctions of literals considered up to commutativity and associativity. We will denote them $ C, C_0, C_1$ etc; the empty clause will be denoted by $\bot$. The cardinality of a clause $C$ is denoted $\size C$.

  Finite sets of clauses, \eg $\{C_1,\ldots,C_n \}$, will be denoted $\phi, \phi_0,$ etc. By $\size \phi$ we denote the sum of the sizes of the clauses in $\phi$.  By $\atm \phi$ we denote the set of literals that appear in $\phi$ or whose negations appear in $\phi$.

  Given a theory $\mathcal{T}$ the system \LKDPLLTh, given in Figure~\ref{fig:lkdpllth}, is an inference system on sequents of the form $\DerDPLLTh{\Delta}{\phi}$, where $\Delta$ is a set of literals (\eg $\{l_1,\ldots,l_n \}$).
\end{definition}

\begin{figure}[h!]
  \[\begin{array}{|c|}
    \upline
    \infer[Split $ 
    where $l\in\atm\phi$,
    $ \Del, \non l \nvDash_\mathcal{T} $
    and $ \Del,l \nvDash_\mathcal{T}] 
    {\DerDPLLTh {\Del} {\phi}}
    {
      {\DerDPLLTh {\Del,\non l} {\phi}}  
      \qquad 
      {\DerDPLLTh{\Del, l} {\phi}}
    } \\\\
    \infer[Empty] {\DerDPLLTh {\Del} {\phi, \bot }} {}
    \qquad
    \infer[Assert $ where $ \Del, \non l \nvDash_\mathcal{T} $ and $ \Del,l \nvDash_\mathcal{T}]
    {\DerDPLLTh {\Del} {\phi, l }} {\DerDPLLTh {\Del,l} {\phi, l }}\\\\
    \infer[Subsume $ where $   \Del,\non l \models_\mathcal{T} ]{\DerDPLLTh {\Del} { \phi,l \vee C}}
    {\DerDPLLTh {\Del}{ \phi} }
    \qquad	
    \infer[Resolve $ where $   \Del,l \models_\mathcal{T} ]{\DerDPLLTh {\Del} { \phi,l \vee C}}
    {\DerDPLLTh {\Del}{ \phi,C} }

    \downline
  \end{array} \]
  \caption{System \LKDPLLTh\ }
  \label{fig:lkdpllth}
\end{figure}


The $Assert$  rule models the fact that every literal occurring as a unit clause in the current clause set must be satisfied for the whole clause set to be satisfied. The $Split$ is mainly used to branch the proof tree from the \DPLL\ rewrite sequence system.This rule corresponds to the decomposition in smaller subproblems of the \DPLL\ method. This rule is the only $don't$ $know$ $non-deterministic$ rule of the calculus.
The $Resolve$ rule removes from a clause all literals whose complement has been asserted (which corresponds to generating the simplified clause by unit resolution and the discarding the clause by backword subsumption). The $Subsume$ rule removes from the clauses that contain an asserted literal( because all of these clause will be satisfied in any model in which the asserted literal is true). To close the branch of a proof tree we use the $empty$ rule is in the calculus just for convenience and could be removed with no loss of completeness. It models the fact that a derivation can be terminated as soon as the empty clause is derived. We do not consider that the model is consistent and satisfiable. 

\begin{definition}[Semantical entailment]
  $\Delta \models _\mathcal {T } C$ is a semantical notion of entailment for a particular theory $\mathcal{T}$, \ie every $\mathcal{T}$-model of $\Delta$ is a $\mathcal{T}$-model of $C$. A theory lemma is a clause $C$ such that $\emptyset \models _\mathcal T C$.
\end{definition}

\begin{lemma}[Weakening 1]
  \label{lwt1}
  The following rule is size-preserving admissible in \LKDPLLTh
  \[
  \iinfer
  {\DerDPLLTh{\Delta}{\phi,C}}
  {\DerDPLLTh{\Delta}{\phi}}
  \]
\end{lemma}

\begin{proof}
  By induction on $\DerDPLLTh{\Delta}{\phi}$.
\end{proof}

\begin{definition}[Consequences]
  For every set $\Del$ of literals $l$, let $\mSat \Del = \{ l | \Del \models_\mathcal{T} l \}$ and $\nSat \phi \Del = \mSat \Del \cap \atm \phi $.
\end{definition}

\begin{remark}
  If $\mSat \Del = \mSat {\Del'}$ then $\Delta \models_\mathcal{T} l$ iff $\Delta' \models_\mathcal{T} l$
\end{remark}

\begin{lemma}[Weakening 2]
  \label{sat2}
  The following rule is size-preserving admissible in \LKDPLLTh
  \[
  \iinfer[\nSat{\phi} \Del \subseteq \nSat {\phi} {\Del'}]
  {\DerDPLLTh{\Delta'}{\phi}}
  {\DerDPLLTh{\Delta}{\phi}}
  \]
\end{lemma}

\begin{proof}
  By induction on the derivation of $\DerDPLLTh{\Delta}{\phi}$: 
  \begin{itemize}
  \item[$Resolve$]
    $\infer[ \Del,l \models_\mathcal{T} ]{\DerDPLLTh {\Del} { \phi,l \vee C}}
    {\DerDPLLTh {\Del}{ \phi,C} }$
    
    We assume  $\nSat {\phi, l\vee C} \Del \subseteq \nSat {\phi, l\vee C} {\Del'}$\\ from which we get  $\nSat {\phi,C} {\Del} \subseteq \nSat {\phi,C} {\Del'}$, so we can apply the induction hypothesis to construct
    $$\infer[ \Del',l \models_\mathcal{T} ]{\DerDPLLTh {\Del'} { \phi,l \vee C}}
    {\DerDPLLTh {\Del'}{ \phi,C} }$$
    The side-condition is a consequence of the assumption $\nSat {\phi, l\vee C} \Del \subseteq \nSat {\phi, l\vee C} {\Del'}$.
  \item[$Subsume$]
    \infer[ \Del,\non l \models_\mathcal{T} ]{\DerDPLLTh {\Del} { \phi,l \vee C}}
    {\DerDPLLTh {\Del}{ \phi} }		
    
    We assume  $\nSat {\phi, l\vee C} \Del \subseteq \nSat {\phi, l\vee C} {\Del'}$ \\from which we get  $\nSat {\phi} {\Del} \subseteq \nSat {\phi} {\Del'}$, so we can apply the induction hypothesis to construct
    $$\infer[ \Del',\non l \models_\mathcal{T} ]
    {\DerDPLLTh {\Del'} { \phi,l \vee C}}
    {\DerDPLLTh {\Del'}{ \phi} }		$$
    The side-condition is a consequence of the assumption $\nSat {\phi,l\vee C} \Del \subseteq \nSat {\phi, l\vee C} {\Del'}$.

  \item [$Assert$] 
    $ \infer[ \Del, \non l \nvDash_\mathcal{T} $ and $ \Del,l \nvDash_\mathcal{T}]
    {\DerDPLLTh {\Del} {\phi, l }}
    {\DerDPLLTh {\Del,l} {\phi, l }} $

    We assume  $\nSat {\phi, l} \Del \subseteq \nSat {\phi, l} {\Del'}$ \\from which we get  $\nSat {\phi,l} {\Del,l} \subseteq \nSat {\phi,l} {\Del',l}$.
    \begin{itemize}
    \item If $\Del' \models_\mathcal{T} l$, then $\mSat {\Del',l} = \mSat{\Del'} $, so we have $\nSat {\phi,l} {\Del,l} \subseteq \nSat {\phi,l} {\Del'}$. The induction hypothesis then gives ${\DerDPLLTh {\Del'} {\phi, l }}$.

    \item If $\Del' \models_\mathcal{T} \non l$, then we construct $$\infer[Resolve]{\DerDPLLTh{\Del'} {\phi,l}}  {\infer[Empty]{\DerDPLLTh{\Del'} {\phi,\bot}} {}} $$

    \item If $\Del' \not\models_\mathcal{T} l$ and $\Del' \not\models_\mathcal{T} \non l$: we first apply the induction hypothesis to get 
      $\DerDPLLTh {\Del',l} {\phi, l }$ and we conclude by constructing
      $$  \infer[ \Del', \non l \nvDash_\mathcal{T} $ and $ \Del',l \nvDash_\mathcal{T}]
      {\DerDPLLTh {\Del'} {\phi, l }} 
      {\DerDPLLTh {\Del',l} {\phi, l }} $$
    \end{itemize}

  \item [$Split$] 

    $\infer[ \Del, \non l \nvDash_\mathcal{T} $ and $ \Del,l \nvDash_\mathcal{T}] {\DerDPLLTh {\Del} {\phi, l \vee C}}
    {{\DerDPLLTh {\Del,\non l} {\phi, l \vee C}}  \quad 
      {\DerDPLLTh{\Del, l} {\phi, l \vee C}}} $

    We assume  $\nSat {\phi, l \vee C} \Del \subseteq \nSat {\phi, l \vee C} {\Del'}$ from which we get both\\ $\nSat {\phi,l \vee C} {\Del,l} \subseteq \nSat {\phi,l \vee C} {\Del',l}$ and  $\nSat {\phi,l \vee C} {\Del,\non l} \subseteq \nSat {\phi,l \vee C} {\Del',\non l}$.

    \begin{itemize}
    \item If $\Del ' \models_\mathcal{T} l$, then $\mSat {\Del'} = \mSat {\Del',l}$, so we have $\nSat {\phi,l \vee C} {\Del, l} \subseteq \nSat {\phi,l \vee C} {\Del'}$. The induction hypothesis then gives $\DerDPLLTh {\Del'} {\phi, l \vee C} $.
    \item If $\Del ' \models_\mathcal{T} \non l$, then $\mSat {\Del'} = \mSat {\Del',\non l}$, so we have $\nSat {\phi,l \vee C} {\Del, \non l} \subseteq \nSat {\phi,l \vee C} {\Del'}$. The induction hypothesis then gives $\DerDPLLTh {\Del'} {\phi, l \vee C} $.
    \item If $\Del ' \not \models_\mathcal{T} l$ and $\Del ' \not\models_\mathcal{T} \non l$: the induction hypothesis on both premises gives ${\DerDPLLTh {\Del',l} {\phi, l \vee C }}$ and ${\DerDPLLTh {\Del',\non l} {\phi, l \vee C }}$, and we can conclude
      $$\infer[\Del ' \not \models_\mathcal{T} l \mbox{ and }\Del ' \not\models_\mathcal{T} \non l] {\DerDPLLTh {\Del'} {\phi, l \vee C}}
      {{\DerDPLLTh {\Del',\non l} {\phi, l \vee C}}  \quad 
	{\DerDPLLTh{\Del', l} {\phi, l \vee C}}} $$
    \end{itemize}
  \item [$Empty$]
    Straightforward. 
  \end{itemize}
\end{proof}

\begin{lemma}[Invertibility of Resolve]
  $Resolve$ is size-preserving invertible in \LKDPLLTh.
\end{lemma}
\begin{proof}
  By induction on the derivation of $\DerDPLLTh{\Delta}{\phi,C\vee l}$ we prove $\DerDPLLTh{\Delta}{\phi,C}$ (with the assumption $\Del,l\models_{\mathcal T}$): 
  \begin{itemize}
  \item [$Resolve$] easily permutes with other instances of $Resolve$ and with instances of $Subsume$.
  \item [$Assert$] 
    The side-condition of the rule guarantees that the literal added to the model, say $l'$, is different from $l$:
    $$ \infer[ \Del, \non {l'} \nvDash_\mathcal{T} $ and $ \Del,l' \nvDash_\mathcal{T}]
    {\DerDPLLTh {\Del} {\phi',l',C\vee l}}
    {\DerDPLLTh {\Del,l'} {\phi',l',C\vee l}} $$

    We can construct
    $$ \infer[ \Del, \non {l'} \nvDash_\mathcal{T} $ and $ \Del,l' \nvDash_\mathcal{T}]
    {\DerDPLLTh {\Del} {\phi',l',C}}
    {\DerDPLLTh {\Del,l'} {\phi',l',C}} $$
    whose premiss is proved by the induction hypothesis.

  \item [$Split$] 
    $\infer[ l'\in\atm{\phi,C\vee l}\mbox{ and }\Del, \non {l'} \nvDash_\mathcal{T} $ and $ \Del,l' \nvDash_\mathcal{T}] 
    {\DerDPLLTh {\Del} {\phi,C\vee l}}
    {
      {\DerDPLLTh {\Del,\non {l'}} {\phi,C\vee l}}  
      \quad 
      {\DerDPLLTh{\Del, l'} {\phi,C\vee l}}
    } $

    We can construct
    $$
    \infer[l'\in\atm{\phi,C}\mbox{ and }\Del, \non {l'} \nvDash_\mathcal{T} \mbox{ and } \Del,l' \nvDash_\mathcal{T}]
    {\DerDPLLTh {\Del} {{\phi,C}}}{
      {\DerDPLLTh {\Del,l'} {{\phi,C}}}
      \quad
      {\DerDPLLTh {\Del,\non {l'}} {{\phi,C}}}
    }
    $$
    whose branches are closed by using the induction hypothesis. The side-condition $l'\in\atm{\phi,C}$ is satisfied because $l\neq l'$.
  \item [$Empty$]
    Straightforward. 
  \end{itemize}
\end{proof}


We now introduce a new system  \LKDPLLThp\  which is an extended version of \LKDPLLTh\ with $Weakening 1$, $Weakening2$ and the \textit{Inverted Resolve}. By the previous lemmas, a sequent derivable in \LKDPLLThp\ is derivable in \LKDPLLTh.

\begin{figure}[h!]
  \[\begin{array}{|c|}
    \upline
    

    \iinfer {\DerDPLLTh {\Del} {\phi, C }} {\DerDPLLTh {\Del} {\phi }}
    \qquad
    
    \iinfer[\nSat{\phi} \Del \subseteq \nSat {\phi} {\Del'}  ]
    {\DerDPLLTh{\Delta'}{\phi}} {\DerDPLLTh{\Delta}{\phi}} 
    
    \qquad
    
    \iinfer[\Del,l \models_\mathcal{T} ]{\DerDPLLTh {\Del}{ \phi,C} }
    {\DerDPLLTh {\Del} { \phi,l \vee C}}
    \downline
  \end{array} \]
  \caption{System \LKDPLLThp\ }
  \label{fig:lkdpllthp}
\end{figure}

\begin{definition}[Size of proof-trees in \LKDPLLThp]
  The size of proof-trees in \LKDPLLThp\ is defined as the size of trees in the usual sense, but not counting the occurences of \emph{Weakening1}, \emph{Weakening2} or the \emph{Inverted Resolve} rules.\footnote{For that reason, dashed lines will be used for the occurences of those inference rules.}
\end{definition}

\begin{remark}
  The size-preserving admissibility results of those three rules in \LKDPLLTh\ entails that a proof-tree in \LKDPLLThp\ of size $n$, can be transformed into a proof-tree in \LKDPLLTh\ of size at most $n$.  
\end{remark}

\begin{lemma}
  \label{lgt}
  If $\Delta \models_{\mathcal{T}}  \neg C$ then there is a proof-tree concluding $\DerDPLLTh\Delta{{C,\phi}}$ of size at most $\size \phi+1$.
\end{lemma}

\begin{proof}
  Here $\Delta \models_{\mathcal{T}} \neg C$ means $C= l_1 \vee\ldots\vee l_n$ and for all $l_i$, $\FA l_i$ $\Delta \models_{\mathcal{T}}  \non l_i$ where i=1,\ldots,n.\\
  We can therefore construct 
  $$
  \Infer[\mbox{Resolve}]{\DerDPLLTh\Delta{{C,\phi}} }
  {\infer[\mbox{Empty}]{\DerDPLLTh\Delta{{\bot,\phi}} }{}}
  $$
\end{proof}

\subsection{\DPLLTh\ with backtracking}

In this section we describe the basic \DPLLTh\ procedure~\cite{Nieuwenhuis06}, and its encoding into \LKDPLLTh.

\begin{definition}[Basic \DPLLTh]
  {\em Models} are defined by the following grammar:
  $$ \Del\recdef \emptyenv \sep \Del,l^d \sep \Del,l$$
  where $l$ ranges over literals, and $l^d$ is an annotated literal called decision literal.

  The basic \DPLLTh\ procedure rewrites {\em states} of the form $\Del \| \phi$, with the following rewriting rules:
  \begin{itemize}
  \item Fail:\\
    $\Del \| \phi,C \Rightarrow$  \unsat,
    \hfill with $\forget{\Del}  \models \neg C$ and there is no decision literal in $\Delta$. 
  \item Decide:\\
    $\Delta  \| \phi \Rightarrow \Delta, l^d   \| \phi$
    \hfill where $l \not \in \Delta$, $\non l \not \in \Delta$, $l \not \in \phi $  or $\non l \not \in \phi$.
  \item
    Backtrack:\\
    $\Delta_1, l^d, \Delta_2 \| \phi, C \Rightarrow \Delta_1, \non l \| \phi, C$
    \hfill if $\forget{\Delta_1, l, \Delta_2} \models \neg C$ and no decision literal is in $\Delta_2$. 
  \item
    Unit propagation:\\ $\Delta \| \phi, C \vee l \Rightarrow \Delta, l \| \phi, C \vee l$
    \hfill where $\forget\Delta \models \neg C$, $l \not \in \Delta$, $\non l \not \in \Delta$. 
  \item 
    Theory Propagate: \\$\Delta \| \phi \Rightarrow \Delta, l \| \phi$
    \hfill where $\forget\Delta \models_\mathcal{T} l$, $l\in\atm\phi$ and $l \not \in \Delta, \non l \not \in \Delta$.
  \end{itemize}
  where $\forget\Delta$ denotes the result of erasing the annotations on decision literals, an operation defined in Figure~\ref{fig:decisioneraseT}.

  \begin{figure}[!h]
    \[
    \begin{array}{|c|}
      \upline
      \begin{array}{ll}
        \forget{ \emptyenv } &\eqdef  \emptyenv\\
        \forget{\Delta,l} &\eqdef \forget\Delta,l\\
        \forget{ \Delta,l^d} &\eqdef \forget\Delta,l
      \end{array}
      \downline
    \end{array}
    \]
    \caption{Erasing annotations}
    \label{fig:decisioneraseT}
  \end{figure}

\end{definition}

We now proceed with the encoding of the basic \DPLLTh\ procedure as the construction of a derivation tree in System \LKDPLLTh. The simulation could be be stated as follows:

If $\Del \| \phi \Rightarrow^* \unsat$ then there is a \LKDPLLTh\ proof of  $\DerDPLLTh {\forget \Delta} \phi $ (\ie there is no $\mathcal{T}$-model of $\phi$ extending $\Delta$).

This is true; however, there is more information in $\Del \| \phi \Rightarrow^* \unsat$ than in $\DerDPLLTh {\forget \Delta} \phi $, because the \DPLLTh\ sequence leading to $\unsat$ also backtracks on decision literals.  This means that not only there is no $\mathcal{T}$- model of $\phi$ extending $\forget \Delta$, but no matter how decision literals of $\Delta$ are changed, there is still no $\mathcal{T}$- model of $\phi$ that can be constructed. This notion is formalised by collecting the backtrack models as follows:

\begin{definition}[Backtrack models]
  In Fig.~\ref{fig:backtrackcollectT} we define the interpretation of a model as a collection (formally, a multiset) of sets of literals.
  \begin{figure}[!h]
    \[
    \begin{array}{|c|}
      \upline
      \begin{array}{ll}
        \backstrict{ \emptyenv } &\eqdef  \emptyset\\
        \backstrict{ \Delta,l } &\eqdef \backstrict \Delta\\
        \backstrict{ \Delta,l^d} &\eqdef \backpoints {\Delta,\non l}\\&\\
        \backpoints\Delta &\eqdef \backstrict\Delta\cup\multiset{\forget\Delta}
      \end{array}
      \downline
    \end{array}
    \]
    \caption{Collecting backtrack points}
    \label{fig:backtrackcollectT}
  \end{figure}
\end{definition}

\begin{remark}
  We have $\forget\Delta \in \backpoints \Delta $ and $\backstrict \Delta  \subseteq \backpoints \Delta$.
\end{remark}

We consider a notion of  a partial proof-tree to step-by-step simulate \DPLLTh\ runs.
\begin{definition}[Partial proof-tree]
  A partial proof-tree in \LKDPLLThp\ is a tree labelled with sequents, whose leaves are tagged as either \emph{open} or \emph{closed}, and such that every node that is not an open leaf is an instance of the \LKDPLLThp\ rules.\footnote{A partial proof-tree that has no open leaf is isomorphic to a derivation in \LKDPLLThp.} 

  A complete proof-tree is a partial proof-tree whose leaves are all closed.

  A partial proof-tree $\pi'$ is an \emph{$n$-extension} of $\pi$ if $\pi'$ is $\pi$ or if $\pi'$ is obtained from $\pi$ by replacing one of its open leaves by a partial proof-tree of size at most $n$ and whose conclusion has the same label as that leaf.
\end{definition}

\begin{definition}[Correspondence between \DPLLTh\ states and partial proof-trees]
  A partial proof-tree $\pi$ \emph{corresponds to} a \DPLLTh\ state $\Del \| \phi$ if the sequents labelling its open leaves form a sub-set of
  $\set{\DerDPLLTh{\Del'}{\phi} \mid \Del'\in\backpoints\Del}$.

  A partial proof-tree $\pi$ corresponds to \unsat\ if it has no open leaf.
\end{definition}

The \DPLLTh\ procedure starts from an initial state \ie $\emptyset \| {\phi}$, to which corresponds the partial proof-tree consisting of one node (both a root and a leaf) labelled with the sequent $\DerDPLLTh{}{\phi}$.

Note that, different partial proof-trees might correspond to the same \DPLLTh\ state, as different \DPLLTh\ runs can lead to that state from various initial \DPLLTh\ states. 
The simulation theorem below expresses the fact that, when \DPLLTh\ rewrites one state to another state, any partial proof-tree corresponding to the formal state can be extended into a partial proof-tree corresponding to the latter state.    

\begin{theorem}
  \label{TTFail}
  If  $ \Del \| \phi \Rightarrow \mathcal S_2$ is a rewrite step of  \DPLLTh\, and if $\pi_1$ corresponds to $\Del \| \phi$ then there is, in \LKDPLLThp, a $\size \phi+1$-extension $\pi_2$ of $\pi_1$ corresponding to $\mathcal S_2$.
\end{theorem}

\begin{proof}
  By case analysis:
  \begin{itemize}

  \item Fail: $\Del \| \phi,C \Rightarrow^*\unsat$
    \hfill with $\forget{\Del}  \models \neg C$ and there is no decision literal in $\Delta$.

    Let $\pi_1$ be a partial proof-tree corresponding to $\Del \| \phi,C$. Since there are no decision literals in $\Del$, $\pi_1$ can have at most one open leaf, labelled by $\DerDPLLTh{\forget \Del}{\phi,C}$.

    We $\size{\phi,C}$+1-extend $\pi_1$ into $\pi_2$ by replacing that leaf by a complete tree deriving $\DerDPLLTh{\forget \Del}{\phi,C}$. We obtain that tree by applying Lemma~\ref{lgt} on the hypothesis $\forget\Del  \models \neg C$. The new tree $\pi_2$ is complete and therefore corresponds to the \unsat\ state of the \DPLLTh\ run.  

  \item Decide: $\Delta  \| \phi \Rightarrow \Delta, l^d   \| \phi$ \hfill where $l \not \in \Delta$, $\non l \not \in \Delta$, $l \in \phi $  or $\non l \in \phi$.\\

    Let $\pi_1$ be a partial proof-tree corresponding to $\Del \| \phi$.
    We $1$-extend it into $\pi_2$ by replacing the open leaf labelled with $\DerDPLLTh{\forget\Del}{\phi}$ (if there is such a leaf) by one of three proof-trees:
    \begin{itemize} 
    \item If $\forget \Del, l \models_\mathcal{T}$, we have $\mSat {\forget \Del}= \mSat {\forget \Del,\non l}$ and we take: 
      \[
      \iinfer[\mbox{Weakening2}]
      {\DerDPLLTh{\forget \Delta}{\phi}}
      {\DerDPLLTh{\forget \Del, \non l}{\phi}}
      \]
      The new open leaves form a sub-set of $\set{\DerDPLLTh{\forget \Delta, \non l}{\phi}}\cup\set{\DerDPLLTh{\Delta'}{\phi}\mid \Delta'\in\backstrict{\Del}}\subseteq\set{\DerDPLLTh{\Delta'}{\phi}\mid \Delta'\in\backpoints{\Del,l^d}}$ (since $\forget \Del, \non l = \forget{\Delta, \non l} \in \backpoints {\Delta, \non l} = \backstrict {\Delta, l^d} \subseteq \backpoints {\Delta, l^d}$) and therefore $\pi_2$ corresponds to $\Delta, l^d   \| \phi$.

    \item If $\forget \Del, \non l \models_\mathcal{T}$, we have $\mSat {\forget \Del}= \mSat {\forget \Del, l}$ and we take
      \[\iinfer[\mbox{Weakening2}]
      {\DerDPLLTh{\forget \Delta}{{\phi}}}
      {\DerDPLLTh{\forget \Del, l}\phi}
      \]
      The new open leaves form a sub-iset of  $\set{\DerDPLLTh{\forget \Delta, l}{\phi}}\cup\set{\DerDPLLTh{\Delta'}{\phi}\mid \Delta'\in\backstrict{\Del}}\subseteq\set{\DerDPLLTh{\Delta'}{\phi}\mid \Delta'\in\backpoints{\Del,l^d}}$ (since $\forget {\Del}, l  =\forget {\Del,l} \in \backpoints {\Del, l^d}$) and therefore $\pi_2$ corresponds to $\Delta, l^d   \| \phi$.

    \item If $\forget \Del, l \not\models_{\mathcal{T}} $ and $\forget \Del, \non l \not\models_{\mathcal{T}} $, we take
      \[\infer[\mbox{Split}] {\DerDPLLTh{\forget \Del}{{\phi}}} 
      {\DerDPLLTh{\forget \Del, l}{{\phi}}   
        \quad    
        {\DerDPLLTh{\forget \Del, \non l} {{\phi}}}}
      \]
      The new open leaves form a sub-set of  $\set{\DerDPLLTh{\forget \Delta, l}{\phi}}\cup\set{\DerDPLLTh{\forget \Delta, \non l}{\phi}}\cup\set{\DerDPLLTh{\Delta'}{\phi}\mid \Delta'\in\backstrict{\Del}}\subseteq\set{\DerDPLLTh{\Delta'}{\phi}\mid \Delta'\in\backpoints{\Del,l^d}}$ and therefore $\pi_2$ corresponds to $\Delta, l^d   \| \phi$. (since $\forget \Del, \non l = \forget{\Delta, \non l} \in \backpoints {\Delta, \non l} = \backstrict {\Delta, l^d} \subseteq \backpoints {\Delta, l^d}$)       
    \end{itemize}

  \item Backtrack: $\Delta_1, l^d, \Delta_2 \| \phi, C \Rightarrow \Delta_1, \non l \| \phi, C$\\
    \strut\hfill if $\forget{\Delta_1, l, \Delta_2} \models \neg C$ and no decision literal is in $\Delta_2$. \\
    
    Let $\pi_1$ be a partial proof-tree corresponding to $\Delta_1, l^d, \Delta_2 \| \phi, C$. Since there are no decision literal in $\Del_2$, $\pi_1$ can have at most one open leaf, labelled with $\DerDPLLTh{\forget {\Delta_1, l^d, \Delta_2} }{\phi,C}$.
    
    We $\size{\phi,C}$+1-extend $\pi_1$ into $\pi_2$ by replacing that leaf by a complete tree deriving $\DerDPLLTh{\forget {\Delta_1, l^d, \Delta_2} }{\phi,C}$. 
    We obtain that partial proof-tree by applying lemma~\ref{lgt} on the assumption $\forget {\Delta_1, l^d, \Delta_2} \models \neg C$.

    The new open leaves form a sub-set of  $\set{\DerDPLLTh{\forget \Delta_1, \non l}{\phi,C}}\cup\set{\DerDPLLTh{\Delta'}{\phi}\mid \Delta'\in\backstrict{\Del_1,l^d, \Del_2}}=\set{\DerDPLLTh{\Delta'}{\phi}\mid \Delta'\in\backstrict{\Del_1,l^d}}\subseteq\set{\DerDPLLTh{\Delta'}{\phi}\mid \Delta'\in\backpoints{\Del_1, \non l}}$ (since $\forget {\Del_1}, \non l  =\forget {\Del_1, \non l} \in \backpoints {\Del_1, \non l}$)  and therefore $\pi_2$ corresponds to $\Delta_1, \non l   \| \phi,C$ state of the \DPLLTh\ run.

  \item Unit propagation : $\Delta \| \phi, C \vee l \Rightarrow \Delta, l \| \phi, C \vee l$
    \hfill where $\forget\Delta \models \neg C$, $l \not \in \Delta$, $\non l \not \in \Delta$.
    
    Let $\pi_1$ be a partial proof-tree corresponding to $\Delta \| \phi, C \vee l$.
    We $\size{\phi, C \vee l}$+1-extend it into $\pi_2$ by replacing the open leaf labelled with $\DerDPLLTh{\forget\Del}{\phi,C \vee l}$ (if there is such a leaf) by one of three proof-trees:
    \begin{itemize} 
    \item If $\forget \Del, \non l \models_\mathcal{T}$, we have $\mSat {\forget \Del}= \mSat {\forget \Del, l}$ and we take: 
      \[
      \iinfer[\mbox{Weakening2}]
      {\DerDPLLTh{\forget \Delta}{\phi}}
      {\DerDPLLTh{\forget \Del, l}{\phi}}
      \]
      The new open leaves form a sub-set of $\set{\DerDPLLTh{\forget \Delta,  l}{\phi,C \vee l}}\cup\set{\DerDPLLTh{\Delta'}{\phi,C \vee l}\mid \Delta'\in\backstrict{\Del}}\subseteq\set{\DerDPLLTh{\Delta'}{\phi,C \vee l}\mid \Delta'\in\backpoints{\Del,l}}$ (since $\forget \Del,  l = \forget{\Delta, l} \in \backpoints {\Delta, l}$) and therefore $\pi_2$ corresponds to $\Delta, l \| \phi, C \vee l$.

    \item If $\forget \Del,  l \models_\mathcal{T}$ then lemma~\ref{lgt} directly provides a partial proof-tree of $\DerDPLLTh{\forget\Del}{\phi,C \vee l}$.  

    \item If $\forget \Del, l \not\models_{\mathcal{T}} $ and $\forget \Del, \non l \not\models_{\mathcal{T}} $, we can construct the following tree:
      \[ 
      \Infer[\mbox{Resolve}]{\DerDPLLTh{\forget\Delta}{{\phi},C\vee l }         }
      { 
        \infer[\mbox{Assert}] {\DerDPLLTh{\forget\Delta}   {{\phi},l}    }
        {
          { \iinfer[\mbox{Inverted Resolve}] {[=]\DerDPLLTh{\forget \Delta,l}{\phi,l}    }
            {\DerDPLLTh{\forget \Delta,l}{\phi,C \vee l}     }
          }
        } 
      } 
      \] 
      
      where the side-conditions of $Resolve$ are provided by the hypothesis $\Del'' \models \neg C$.

      The new open leaves form a sub-set of  $\set{\DerDPLLTh{\forget \Delta, l}{\phi,C \vee l}}\cup\set{\DerDPLLTh{\Delta'}{\phi}\mid \Delta'\in\backstrict{\Del}} \subseteq \set{\DerDPLLTh{\Delta'}{\phi}\mid \Delta'\in\backpoints{\Del,l}} $ and therefore $\pi_2$ corresponds to $\Delta, l \| \phi,C \vee l$. (since $\forget \Del, l = \forget{\Delta, l} \in  \backpoints {\Delta, l}$)       
    \end{itemize}

  \item Theory Propagate: $\Delta \| \phi \Rightarrow \Delta, l \| \phi$
    \hfill where $\forget\Delta \models_\mathcal{T} l$, $l\in\atm\phi$ and $l \not \in \Delta, \non l \not \in \Delta$.\\
    
    Let $\pi_1$ be a partial proof-tree corresponding to $\Delta \| \phi$.
    We $1$-extend it into $\pi_2$ by replacing the open leaf labelled with $\DerDPLLTh{\forget\Del}{\phi}$ by the following proof-tree :
    \[
    \iinfer[\mbox{Weakening2}]
    {\DerDPLLTh{\forget \Delta}{\phi}}
    {\DerDPLLTh{\forget \Del, l}{\phi}}
    \]
    The new open leaves form a sub-set of $\set{\DerDPLLTh{\forget \Delta,  l}{\phi}}\cup\set{\DerDPLLTh{\Delta'}{\phi}\mid \Delta'\in\backstrict{\Del}}\subseteq\set{\DerDPLLTh{\Delta'}{\phi}\mid \Delta'\in\backpoints{\Del,l}}$ (since $\forget \Del, l = \forget{\Delta, l}  \subseteq \backpoints {\Delta, l}$) and therefore $\pi_2$ corresponds to $\Delta, l \| \phi$.

  \end{itemize}
\end{proof}

\begin{corollary} 
  \label{clkdpllth}
  \LKDPLLTh\ is complete, \ie if $ \phi \models_{\mathcal{T}} $ then $\DerDPLLTh{}{\phi}$.
\end{corollary}
\begin{proof}
  By completeness of basic \DPLLTh\ and Theorem~\ref{TTFail}.
\end{proof}

\subsection{\DPLLTh\ with backjumping and Lemma learning}
We now consider a more advanced version of \DPLLTh, which involves backjumping and lemma learning features, and which we denote \DPLLbjTh. \DPLLbjTh\ extends basic \DPLLTh\ with the rules known as $\mathcal{T}$-Backjump, $\mathcal T$-Learn, $\mathcal T$-Forget, and Restart~\cite{Nieuwenhuis06}. Those rules drastically increase the efficiency of SMT-solvers.

\begin{itemize}
\item[$\mathcal{T}$-Backjump:]  
  $\Delta_1, l^d, \Delta_2 \| \phi, C \Rightarrow \Delta_1, l_{bj} \| \phi, C$ with
  \begin{enumerate}
  \item $\forget{\Delta_1, l^d,\Delta_2} \models \neg C$. 
  \item $\forget {\Delta_1} \models \neg C'$
  \item $\phi, C \models_{\mathcal{T}} C' \vee l_{bj}$
  \item $l_{bj} \not \in \Delta_1$, $\non l_{bj} \not \in \Delta_1$ and $l_{bj} \in \atm{\phi, \Delta_1,l^d,\Delta_2}$.
  \end{enumerate}
  for some clause $C'$ such that $\atm{C'}\subset\atm{\phi,C}$.
\item[$\mathcal T$-Learn:] $\Delta \| \phi \Rightarrow \Delta \| \phi,C$ if $\atm C\subseteq\atm{\phi,\Delta}$ and $\phi \models_{\mathcal{T}} C$.
\item[$\mathcal T$-Forget:] $\Delta \| \phi , C \Rightarrow \Delta \| \phi$ if $\phi \models_{\mathcal{T}} C$.
\item[Restart:] $\Delta \| \phi \Rightarrow \emptyset \| \phi$. 
\end{itemize}

In order to simulate those extra rules in \LKDPLLTh, we need to extend \LKDPLLTh\ with a cut rule as follows:

\begin{definition}[\LKDPLLTh\ with cut]
  System \LKDPLLThc\ is obtained by extending system \LKDPLLThp\ with the following cut-rule:
  \[\begin{array}{c}
    \infer[Cut $ where $ C = \non l_1,\ldots,\non l_n]  { \DerDPLLTh {\Del} {\phi} } 
    { \DerDPLLTh {\Del} {\phi, l_1,\ldots,l_n}
      \quad  \DerDPLLTh{\Del}  {\phi,C}
    }
  \end{array} \]	

  We define the size of proof-trees in \LKDPLLThc\ as we did for \LKDPLLThp\ (ignoring \emph{Weakening1}, \emph{Weakening2} or the \emph{Inverted Resolve}), but also ignoring the left-branch of the cut-rules.\footnote{As we shall see in the simulation theorem, this definition mimicks the fact that the length of \DPLLTh\ sequences is a complexity measure that ignores the cost of checking the side-conditions.} 
\end{definition}

\begin{definition}[$n,\phi,\mathcal{S}$-sync action]
  $\pi_\phi$ is a $n,\phi,\mathcal{S}$-sync action if it is a function that maps every model $\Del \in \mathcal{S}$ to a partial proof-tree of size at most $n$ and concluding $\DerDPLLTh{\Del}{\phi}$.
\end{definition}

\begin{definition}[Parallel $n$-extension of partial proof-trees]
  $\pi_2$ is a \emph{parallel n-extension} of $\pi_1$ according to $\pi_\phi$ if $\pi_\phi$ is a $n,\phi,\mathcal{S}$-sync action and if $\pi_2$ is obtained from $\pi_1$ by replacing all the open leaves of $\pi_1$ labelled by sequents of the form $\DerDPLLTh{\Del}{\phi}$ (where $\Del\in \mathcal{S}$) by $\pi_\phi(\Del)$.
\end{definition}

\begin{theorem} 
  \label{TMFail}
  If  $\Delta \| \phi\Rightarrow_{\mbox{\DPLLbjTh\ }} \mathcal S_2$ and $\pi_1$ corresponds to $\Del \| \phi$, there is parallel $\size {\phi} +3$-extension $\pi_2$ of $\pi_1$ (according to some $\pi_\phi$) such that $\pi_2$ corresponds to $\mathcal S_2$.  
\end{theorem}

\begin{proof} 
  Since \LKDPLLTh\ is a sub-system of \LKDPLLThc, we only need to simulate (in \LKDPLLThc) the new rules.\\
  
  \begin{itemize}
  \item[$\mathcal{T}$-Backjump:] 
    $\Delta_1, l^d, \Delta_2 \| \phi, C \Rightarrow \Delta_1, l_{bj} \| \phi, C$ with
    \begin{enumerate}
    \item $\forget{\Delta_1, l^d,\Delta_2} \models \neg C$. 
    \item $\forget {\Delta_1} \models \neg C'$
    \item $\phi, C \models_{\mathcal{T}} C' \vee l_{bj}$
    \item $l_{bj} \not \in \Delta_1$, $\non l_{bj} \not \in \Delta_1$ and $l_{bj} \in \atm{\phi, \Delta_1,l^d,\Delta_2}$.
    \end{enumerate}

    Let $\pi_1$ be a partial proof-tree corresponding to $\Delta_1, l^d, \Delta_2 \| \phi, C$. We have to build a $\pi_2$ that corresponds to $\Delta_1, l_{bj} \| \phi, C$ in the \DPLLbjTh\ run. This means that the open leaves of $\pi_2$ should be labelled with sequents of the form $\DerDPLLTh{\Del'}{\phi,C}$ where $\Del'\in \backpoints{\Del_1,l_{bj}}$ .
    
    Let $\mathcal{S} = \backpoints {\Del_1, l^d, \Del_2}  \backslash \backstrict {\Del_1}$ and $\pi_\phi$ be the $\size{\phi,C}$+3, $\phi,C$,$\mathcal{S}$-sync action that maps every $\Del \in \mathcal{S}$ to

    \[\begin{array}{c}
      \infer[Weakening 2]{\DerDPLLTh {\Delta}{\phi, C}}    
      { \infer[cut]{\DerDPLLTh {\forget{\Delta_1}}{\phi, C}}    
        {  
          \iinfer[Weakening 2]{\DerDPLLTh{\forget{\Delta_1}}{\phi,C,\neg C',\non l_{bj} }}{\DerDPLLTh{}{\phi,C,\neg C',\non l_{bj} }}
          \quad 
          \Infer[Resolve]{\DerDPLLTh{\forget {\Delta_1}}{\phi,C,C'\vee{l_{bj}} }}
          {
            \infer[Assert]{\DerDPLLTh{\forget{\Delta_1}}{\phi,C,{l_{bj}} }}
            {
              \infer[Subsume]{\DerDPLLTh{\forget{\Delta_1},{l_{bj}}}{\phi,C,{l_{bj}} }}
              {
                \DerDPLLTh{\forget{\Delta_1},{l_{bj}}}{\phi,C}
              }
            }
          }
        } }
    \end{array}
    \]

    It is a valid partial proof-tree because $\Del \in \mathcal{S}$ entails $\forget {\Del_1} \subseteq \Del$ and therefore $\nSat{\phi}{\forget {\Del_1}} \subseteq \nSat{\phi}{\Del}$.
    The left branch is closed by assumption (3) and the completeness of \LKDPLLTh\ on $\phi, C,\neg C',\non l_{bj}\models_{\mathcal{T}}$ (Corollary~\ref{clkdpllth}). We cannot anticipate the size of the proof-tree closing that branch, and we therefore ignore that proof-tree to compute the size of the whole tree, just as the length of the \DPLLTh\ run ignores the cost of checking $\phi, C\models_{\mathcal{T}} C'\vee l_{bj}$.\\

    Let $\pi_2$ be the \emph{parallel} $\size{\phi,C}+3$-extension of $\pi_1$ according to $\pi_\phi$.  The new open leaves form a sub-set of $\set{\DerDPLLTh{\forget {\Delta_1}, l_{bj}}{\phi,C}}\cup\set{\DerDPLLTh{\Delta'}{\phi}\mid \Delta'\in\backstrict{\Del_1}} \subseteq \set{\DerDPLLTh{\Delta'}{\phi}\mid \Delta'\in\backpoints{\Del_1,l_{bj}}}$ (since $\forget {\Del_1}, l_{bj} = \forget{\Delta_1, l_{bj}} \in \backpoints {\Del_1,l_{bj}}$ and $\backstrict{{\Del_1},l_{bj}}=\backstrict{\Del_1}$ ) and therefore $\pi_2$ corresponds to $\Delta_1, l_{bj}   \| \phi,C$.

  \item[$\mathcal T$-Learn:] $\Delta \| \phi \Rightarrow \Delta \| \phi,C$ if each atom of $C$ occurs in $\phi$ or in $\Delta$ and $\phi \models_{\mathcal{T}} C$.

    Let $\pi_1$ be a partial proof-tree corresponding to $\Delta \| \phi$. We have to build a $\pi_2$ that corresponds to $\Delta \| \phi,C$ in the \DPLLbjTh\ run. This means that the open leaves of $\pi_2$ should be labelled with sequents of the form $\DerDPLLTh{\Del'}{\phi,C}$ where $\Del'\in \backpoints{\Del}$ .
    
    Let $\mathcal{S} = \backpoints \Del$ and $\pi_\phi$ be the $\size{\phi}$,$\phi$,$\mathcal{S}$-sync actionthat maps every $\Del \in \mathcal{S}$ to:

    %

    \[ 
    \infer[cut]{  \DerDPLLTh {\Del} {\phi}   }
    { 
      \iinfer[Weakening2]{\DerDPLLTh {\Del} {{\phi, \neg C} }    }
      {\DerDPLLTh {} {{\phi, \neg C} }    }
      \quad 
      {\DerDPLLTh {\forget\Del} {{\phi,  C} }    } 
    } 
    \]

    The left branch of the cut is closed by assumption and completeness of \LKDPLLTh\ on $\phi,\neg C\models_{\mathcal{T}}$ (Corollary~\ref{clkdpllth}). We cannot anticipate the size of the proof-tree closing that branch, and we therefore ignore that proof-tree to compute the size of the whole tree, just as the length of the \DPLLTh\ run ignores the cost of checking $\phi \models_{\mathcal{T}} C$. \\

    Let $\pi_2$ be the \emph{parallel} $\size{\phi}$-extension of $\pi_1$ according to $\pi_\phi$.  The new open leaves form a sub-set of $\set{\DerDPLLTh{\Delta'}{\phi,C}\mid \Delta'\in\backpoints{\Del}}$ and therefore $\pi_2$ corresponds to $\Del \| \phi,C$.  .

    
  \item[$\mathcal T$-Forget:] $\Delta \| \phi , C \Rightarrow \Delta \| \phi$ if $\phi \models_{\mathcal{T}} C$.\\

    Let $\pi_1$ be a partial proof-tree corresponding to $\Delta \| \phi,C$. We have to build a $\pi_2$ that corresponds to $\Delta \| \phi$ in the \DPLLbjTh\ run. This means that the open leaves of $\pi_2$ should be labelled with sequents of the form $\DerDPLLTh{\Del'}{\phi}$ where $\Del'\in \backpoints{\Del}$ .
    
    Let $\mathcal{S} = \backpoints {\Del}$ and $\pi_\phi$ be the $1,\phi,C$,$\mathcal{S}$-sync action that maps every $\Del' \in \mathcal{S}$ to

    \[
    \iinfer[Weakening1]
    {\DerDPLLTh{ \Delta'}{\phi,C}}
    {\DerDPLLTh{\Delta'}{\phi}}
    \]

    Let $\pi_2$ be the \emph{parallel} $1$-extension of $\pi_1$ according to $\pi_\phi$. The new open leaves form a sub-set of $\set{\DerDPLLTh{\Delta'}{\phi}\mid \Delta'\in\backpoints{\Del}}$ and therefore $\pi_2$ corresponds to $\Del \| \phi$.

  \item[Restart:] $\Delta \| \phi \Rightarrow \emptyset \| \phi$.\\

    Let $\pi_1$ be a partial proof-tree corresponding to $\Delta \| \phi$. We have to build a $\pi_2$ that corresponds to $\emptyset\| \phi$ in the \DPLLbjTh\ run. This means that the open leaves of $\pi_2$ should be labelled with sequents of the form $\DerDPLLTh{}{\phi}$ 
    
    Let $\mathcal{S} = \backpoints {\Del}$ and $\pi_\phi$ be the $1,\phi$,$\mathcal{S}$-sync action that maps every $\Del' \in \mathcal{S}$ to:
    \[
    \iinfer[Weakening2]
    {\DerDPLLTh{\Delta'}{\phi}}
    {\DerDPLLTh{}{\phi}}
    \]

    Let $\pi_2$ be the \emph{parallel} $1$-extension of $\pi_1$ according to $\pi_\phi$. The new open leaves form a sub-set of $\set{\DerDPLLTh{}{\phi}}$ and therefore $\pi_2$ corresponds to $\emptyset \| \phi$. 

  \end{itemize}
\end{proof} 
%

\newpage
\section{Encoding \LKDPLLTh\ in \LKThp}

\subsection{Preliminaries: System \LKThp}

In this section we introduce (the propositional fragment of) system \LKThp.


\begin{definition}[Formulae, negation]\strut
  \label{def:seminconsistency}
  The formulae of \LKThp\ are given by the following grammar:
  \[
  \begin{array}{lll}
    \mbox{Formulae }&A,B,\ldots&\recdef l\mid A\andP B\mid A\orP B\mid A\andN B\mid A\orN B
  \end{array}
  \]
  where $l$ ranges over literals.

  Let $\mathcal P$ be a set of literals declared to be \emph{positive}, while their negations, required to not be in $\mathcal P$, are declared to be \emph{negative}.
  Given such a set $\mathcal P$, we define positive formulae and negative formulae as the formulae generated by the following grammars:
  \[
  \begin{array}{lll}
    \mbox{positive formulae }&P,\ldots&\recdef p\mid A\andP B\mid A\orP B\\
    \mbox{negative formulae }&N,\ldots&\recdef \non p\mid A\andN B\mid A\orN B\\
  \end{array}
  \]
  where $p$ ranges over $\mathcal P$.

  Negation is recursively extended into a involutive map from formulae to formulae as follows:
  \[
  \begin{array}{|ll|ll|}
    \hline
    \non{(A\andP B)}&\eqdef \non A \orN\non B&\non{(A\andN B)}&\eqdef \non A \orP\non B    \\
    \non{(A\orP B)}&\eqdef \non A \andN\non B&\non{(A\orN B)}&\eqdef \non A \andP\non B    \\
    \hline
  \end{array}
  \]
\end{definition}

\begin{definition}[System \LKThp]

  The sequent calculus \LKThp\ has two kinds of sequents:
  \begin{centre}
    \begin{tabular}{ l c r }
      $\DerPosTh \Gamma P {}{}$ &  &  where \emph{P} is in the \emph{focus} of the sequent\\
      $\DerNegTh \Gamma {\Gam'}{}{}$ &  &
    \end{tabular}
  \end{centre}
  Its rules are given in Figure~\ref{fig:LKThp}.

  $\Theory\Delta$ is the call to the decision procedure on the conjunction of all atomic formulae within $\Delta$. It holds if the procedure returns \textsf{UNSAT}.
\end{definition}


\begin{figure}[!h]
  $$
  \begin{array}{|c|}
    \upline
    \infer{\DerPosTh{\Gamma}{A\andP B}{}{\mathcal{P}}}
    {\DerPosTh{\Gamma}{A}
      {}{\mathcal P} \qquad \DerPosTh{\Gamma}{B}{}{\mathcal{P}}}
    \qquad
    \infer{\DerPosTh{\Gamma}{A_1\orP A_2}{}{\mathcal{P}}}
    {\DerPosTh{\Gamma}{A_i}{}{\mathcal{P}}}
    \\\\
    \infer{\DerPosTh{\Gamma,p}{p} {} {\mathcal{P},p}} {\strut}
    \qquad
    \infer{\DerPosTh{\Gamma  }{p} {} {\mathcal{P},p}} {\Theory{\Gam,\non p}}
    \\\\
    \infer[N\mbox{ negative}]{\DerPosTh {\Gam} {N} {} {\mathcal{P}}}
    {\DerNegTh {\Gam} {N} {} {\mathcal{P}}}
    \midline[15]
    \infer{\DerNegTh{\Gamma}{A\andN B,\Delta} {} {\mathcal{P}}}
    {\DerNegTh{\Gamma}{A,\Delta} {} {\mathcal{P}} 
      \qquad \DerNegTh{\Gamma}{B,\Delta} {} {\mathcal{P}}}
    \qquad
    \infer{\DerNegTh {\Gamma} {A_1\orN A_2,\Delta} {} {\mathcal{P}}}
    {\DerNegTh {\Gamma} {A_1,A_2,\Delta} {} {\mathcal{P}}}
    \qquad
    \infer[A\mbox{ positive or atom}]{\DerNegTh \Gam {A,\Delta} {} {\mathcal{P}}} 
    {\DerNegTh {\Gam,\non A} {\Delta} {} {\mathcal{P}}}
    \midline[15]
    \infer{\DerNegTh {\Gam} {} {O,l}{\mathcal{P}}}
    {\DerNegTh{\Gam} {}{O}{\mathcal{P},l}}
    \qquad
    \infer[\begin{array}l P \mbox{ positive}
    \end{array}]
    {\DerNegTh {\Gam,\non P} {}{} {\mathcal{P}}} 
    {\DerPosTh {\Gam,\non P} {P} {} {\mathcal{P}}}
    \qquad
    \infer{\DerNegTh {\Gam} {}{} {\mathcal{P}}}{\Theory \Gam}
    \downline
  \end{array}
  $$
  \caption{System \LKThp}
  \label{fig:LKThp}
\end{figure}

We also consider two cut-rules. The analytic cut:
$$
\infer
{\DerNegTh {\Gam} {} {O}{\mathcal{P}}} 
{\DerNegTh{\Gam,l} {}{O}{\mathcal{P}}\quad \DerNegTh{\Gam,\non l} {}{O}{\mathcal{P}}}
$$
with the condition that $l$ appears in $\Gam$.

The general cut:
\[ 
\infer{\DerNegTh {\Gam} {}{} {\mathcal P}}
{
  \DerNegTh {\Gam, l_1,\ldots,l_n} {}{} {\mathcal P}
  \quad
  \DerNegTh {\Gam, (\non {l_1}\orN\ldots\orN \non {l_n})} {}{} {\mathcal P}
} 
\]

\subsection{Simulation}

We now encode \LKDPLLTh\ in \LKThp.

The main gap between \LKDPLLTh\ (or even \DPLLTh) and a sequent calculus such as \LKThp\ is the fact that the structures handled by the former are very flexible (\eg clauses are multisets of literals), while sequent calculus implements a root-first decomposition of formulae trees.

Clauses in \DPLLTh\ (and in \LKDPLLTh) are disjunctions considered modulo associativity and commutativity.
The way we encode them as formulae of sequent calculus is as follows: a clause $C$ will be represented by a formula $C'$ which is a disjunctive tree whose leaves contain at least all the literals of $C$ but also other literals that we can consider as garbage.

Of course, one could fear that the presence of garbage parts within $C'$ degrades the efficiency of proof-search when simulating \DPLLTh.
This garbage comes from the original clauses at the start of the \DPLLTh\ rewriting sequence, which might have been simplified in later steps of \DPLLTh\ but which remain unchanged in sequent calculus. The size of the garbage is therefore smaller than the size of the original problem. We ensure that the inspection, by the proof-search process, of the garbage in $\non {C'}$, takes no more inference steps than the size of the garbage itself (the waste of time is linear in the size of the garbage). In order to ensure this, we use polarities and the focusing properties of \LKThp: the garbage literals in $C'$ must be negative atoms that are negated in the model/context.

\begin{definition}[$\mathcal{P}$-correspondence]Let $\mathcal{P}$ be a multiset of literals.
  \begin{itemize}
  \item
    \label{def-dpl1}
    A formula $C'$ $\mathcal{P}$-corresponds to a clause $C$ (in system \LKDPLLTh), where $C= l_{1} \vee\ldots\vee l_{p}$, if $C'=l'_{1} \orN \ldots \orN {l'_{p'}}$ with $\{l_{j}\}_{j=1\ldots p} \subseteq \{l'_{j}\}_{j=1\ldots p'}$ and for any $l \in \{l'_{j}\}_{j=1\ldots p'} \backslash \{l_{j}\}_{j=1\ldots p} $, $\non l \in \mathcal{P}$ .
  \item
    \label{def-dpl2}
    A \LKThp\ sequent  $\DerNegTh {\Delta,C'{_1},\ldots,C'{_m}} {}{} {\mathcal{P}}$ corresponds to a \LKDPLLTh\ sequent $\DerDPLLTh{\Delta} {C_1,\ldots,C_m}$, if $C'_i$ $\mathcal{P}$-corresponds to $C_i$ and for all $l \in \mathcal{P}$, $\Delta \models_\mathcal{T} l$. 
  \end{itemize}     

\end{definition}

\begin{lemma}  If $C'$ $\mathcal{P}$-corresponds to $C$, then $C'$ also $(\mathcal{P},l)$-corresponds to $C$.
\end{lemma}
\begin{proof}
  Straightforward.
\end{proof}

\begin{theorem}
  Assume $\infer{\mathcal S}{\mathcal{S}_i}$ is a rule of \LKDPLLTh. For every \LKThp\ sequent $\mathcal{S'}$ that corresponds to $\mathcal{S}$, there exist a partial proof-tree in \LKThp\:
  \begin{itemize}
  \item whose open leaves $(\mathcal {S}'_i)$ are such that $\forall i$, $\mathcal{S}'_i$ corresponds to $\mathcal{S}_i$ and 
  \item  whose size is smaller than size $(\mathcal{S}') + 4$.
  \end{itemize}
  
\end{theorem}

\begin{proof} By case analysis:
  \begin{itemize}
  \item Split: 
    \[  \infer[$ where $l\in\atm\phi,
    \Del, \non l \nvDash_\mathcal{T} 
    $ and $ \Del,l \nvDash_\mathcal{T}] 
    {\DerDPLLTh {\Del} {\phi}}
    {
      {\DerDPLLTh {\Del,\non l} {\phi}}  \quad 
      {\DerDPLLTh {\Del, l} {\phi}}
    }
    \]
    Assume that $\DerDPLLThD{\Delta,\phi'}$ corresponds to 
    $\DerDPLLTh{\Del} {\phi}$ (\ie\ $\phi'=C'_1,\ldots , C'_n$ and $\phi=C_1,\ldots , C_n$ with $C'_i$ $\mathcal{P}$-corresponding to $C_i$ for $i=1\ldots n$).

    We build in \LKThp\ the following derivation that uses an analytic cut:
    \[ 
    \infer{\DerNegTh {\Delta,\phi'} {}{} {\Delta_0}}
    { 
      {\DerNegTh {\Delta,\non l,\phi'} {}{} {\Delta_0}} 
      \quad 
      {\DerNegTh {\Delta,l,\phi'} {}{} {\Delta_0}}
    } 
    \]

    and 
    $\DerNegTh {\Delta,\non l,\phi'} {}{} {\Delta_0}$ $\mathcal{P}$-corresponds to ${\DerDPLLTh {\Del,\non l} {\phi}}$ and 
    $\DerNegTh {\Delta,l,\phi'} {}{} {\Delta,_0}$ $\mathcal{P}$-corresponds to ${\DerDPLLTh{\Del, l} {\phi}}$.
    
  \item Assert:

    \[
    \infer[\Del, \non l \nvDash_\mathcal{T} $ and $ \Del,l \nvDash_\mathcal{T}]
    {\DerDPLLTh {\Del} {\phi, l }} {\DerDPLLTh {\Del,l} {\phi, l }}
    \]
    Assume that $\DerNeg {\Delta,\phi',C'} {}{} {\mathcal{P}}$ corresponds to 
    $\DerDPLLTh {\Del} {\phi,l }$. (\ie $\phi'=C'_1,\ldots , C'_n$ and $\phi=C_1,\ldots , C_n$ with $C'_i$ $\mathcal{P}$-corresponding to $C_i$ for $i=1\ldots n$, and $C'$ $\mathcal{P}$-corresponds to $l$, that is to say $C'=\vee^p_{i=1}l_i$  where $l=l_{i_0}$ for some $i_0 \in 1\ldots n$)

    We build in \LKp\ the following derivation:
    \[ 
    \infer{\DerNegTh {\Delta,\phi',C'} {}{} {\mathcal{P}} }
    { \infer{\DerNegTh {\Delta,\phi',C'} {}{} {\mathcal{P}, l_{i_0}} }
      {
        \prooftree       
        { \infer[i\neq i_0]{\DerPosTh {\Delta,\phi',C'} {\non l_{i}} {} {\mathcal{P},  l_{i_0} } }{\Theory{\Delta,\phi',C', l_{i}}}			\quad 
          \infer{\DerPosTh {\Delta,\phi',C'} { \non {l_{i_0}} } {} {\mathcal{P},  l_{i_0} } }
          { \infer {\DerNegTh{\Delta,\phi',C'}  {\non {l_{i_0}} } {} {\mathcal{P},  l_{i_0} } }
            {\DerNegTh {l_{i_0},\Delta,\phi',C'} {} {} {\mathcal{P},  l_{i_0} } } 
          }
        }
        \using \andP
        \proofdotseparation=1.2ex
        \proofdotnumber=4             .
        \leadsto
        \DerPosTh {\Delta,\phi',C'} {\non {C'}}{} {\mathcal{P}, l_{i_0}  }
        \endprooftree
      } }
    \]
    
    For $i\neq i_0$, $\non l_{i}\in \Delta_0$, so it is positive and we can use an axiom (remember that $\Delta\models\non l_{i}$).

  \item Empty$_\mathcal{T}$:

    \[
    \infer{\DerDPLLTh {\Del} {\phi,\bot }}  {} 
    \]
    Assume that $\DerNeg {\Delta,\phi',C'} {}{} {\mathcal{P}}$ corresponds to 
    $\DerDPLLTh {\Del} {\phi,\bot}$ (\ie $C'$ $\mathcal{P}$-corresponds to $\bot$, $\phi'=C'_1,\ldots , C'_n$ and $\phi=C_1,\ldots ,C_n$ with $C'_i$, $\mathcal{P}$-corresponding to $C_i$ for $i=1\ldots n$).

    We build in \LKp\ the following derivation:

    \[ 
    \infer{\DerNegTh {\Delta,\phi',C'} {}{} {\mathcal{P}}}
    {
      \prooftree       
      \infer{\DerPosTh {\Delta,\phi',C'} {\non l_i}{} {\Delta_0}} {\Theory{\Delta,\phi',C', l_{i}}}
      \using \andP
      \proofdotseparation=1.2ex
      \proofdotnumber=4             .
      \leadsto
      \DerPosTh {\Delta,\phi',C'} {\non {C'}}{} {\mathcal{P}, l_{i_0}  }
      \endprooftree
    } 
    \]

    Again, $\non l_{i}\in \Delta_0$, so it is positive and we can use an axiom (remember that $\Delta\models\non l_{i}$).

    %
  \item Resolve:
    \[ 
    \infer[ \Delta, l \models _\mathcal{T}]{\DerDPLLTh {\Del} { \phi, l \vee C}}
    {\DerDPLLTh {\Del}{ \phi,C} } \] 

    Assume that $\DerNeg {\Delta,\phi',C'} {}{} {\mathcal{P}}$ corresponds to 
    $\DerDPLLTh {\Del} { \phi, l \vee C}$ (\ie\ $C'$ $\mathcal{P}$-corresponds to $ l\vee C$, $\phi'=C'_1,\ldots , C'_n$ and $\phi=C_1,\ldots ,C_n$ with $C'_i$ $\mathcal{P}$-corresponding to $C_i$ for $i=1\ldots n$).
    We build in \LKThp\ the following derivation
    \[ 
    \infer[ pol]{\DerNegTh {\Delta,\phi',C'} {}{} {\mathcal{P}}} 
    {\DerNegTh {\Delta,\phi',C'} {}{} {\mathcal{P}, \non l} } \]

    It suffices to notice that $\DerNeg {\Delta,\phi',C'} {}{} {\mathcal{P},\non l}$ corresponds to $\DerDPLLTh {\Del}{ \phi,C}$. 

  \item Subsume: 
    \[ 
    \infer[ \Delta, \non l \models _\mathcal{T}]{\DerDPLLTh {\Del} { \phi, l \vee C}}
    {\DerDPLLTh {\Del}{ \phi} } \] 
    Assume that $\DerNeg {\Delta,\phi',C'} {}{} {\mathcal{P}}$ corresponds to $\DerDPLLTh {\Del} { \phi, l \vee C}$ (\ie\ $C'$ $\mathcal{P}$-corresponds to $ l\vee C$, $\phi'=C'_1,\ldots , C'_n$ and $\phi=C_1,\ldots ,C_n$ with $C'_i$ $\mathcal{P}$-corresponding to $C_i$ for $i=1\ldots n$). 

    %
    %
  \item Cut: If we want to simulate \DPLLTh\ with backjump, we need to encode the cut rule of \LKDPLLc.
    \[
    \infer[C = \non l_1\vee\ldots\vee\non l_n]{\DerDPLL {\Del} {\phi}}
    {{\DerDPLLTh {\Del} {\phi, l_1,\ldots,l_n}}  \quad 
      {\DerDPLLTh{\Del} {\phi,C}}}
    \]	
    Assume that $\DerNegTh {\Delta,\phi'} {}{} {\mathcal{P}}$ corresponds to 
    $\DerDPLLTh {\Del} { \phi}$ (\ie\ $\phi'=C'_1,\ldots , C'_n$ and $\phi=C_1,\ldots ,C_n$ with $C'_i$ $\mathcal{P}$-corresponding to $C_i$ for $i=1\ldots n$).

    We build in \LKThp\ the following derivation that uses a general cut:

    \[ 
    \infer[cut]{\DerNeg {\Delta,\phi'} {}{} {\mathcal{P}}}
    {
      \DerNeg {\Delta,\phi', l_1,\ldots,l_n} {}{} {\mathcal{P}}
      \quad
      \DerNeg {\Delta,\phi', (\non {l_1}\orN\ldots\orN \non {l_n})} {}{} {\mathcal{P}}
    } 
    \]
    Clearly, $\DerNegTh {\Delta,\phi', l_1,\ldots,l_n} {}{} {\mathcal{P}}$ corresponds to ${\DerDPLLTh {\Del} {\phi, l_1,\ldots,l_n}}$ and 
    $\DerNegTh {\Delta,\phi', (\non {l_1}\orN\ldots\orN \non {l_n})} {}{} {\mathcal{P}}$ corresponds to ${\DerDPLLTh{\Del} {\phi,C}}$.
  \end{itemize}
\end{proof}


\bibliographystyle{Common/good}
\bibliography{Common/abbrev-short,Common/Main,Common/crossrefs}

\end{document}